\def\year{2019}\relax
\pdfoutput=1
\documentclass[letterpaper]{article} 
\usepackage{aaai19}  
\usepackage{times}  
\usepackage{helvet}  
\usepackage{courier}  
\usepackage{url}  
\usepackage{graphicx}  

\usepackage{amsfonts}
\usepackage{amsmath}

\usepackage{amssymb}
\usepackage[vlined,linesnumbered,tworuled,algo2e]{algorithm2e}
\SetAlFnt{\footnotesize}
\SetAlCapFnt{\footnotesize}
\SetAlCapNameFnt{\footnotesize}
\SetVlineSkip{0pt}
\SetAlCapHSkip{0pt}

\usepackage{tikz}
\usepackage{todonotes}
\usepackage{color}
\usepackage{xspace}

\usepackage{amsthm}

\usepackage{multirow}
\newtheorem{theorem}{Theorem}
\newtheorem{lemma}{Lemma}

\newcommand{\MCI}{{\textsc{MCI}}\xspace}
\newcommand{\MCPP}{{\textsc{MCP}}\xspace}
\newcommand{\MSC}{{\textsc{MC}}\xspace}

\newcommand{\DkS}{\textsc{DkS}\xspace}

\newcommand{\Greedy}{\textsc{Greedy}}
\newcommand{\Random}{\textsc{Random }}
\SetKw{Continue}{continue}

\title{Coverage Centrality Maximization in Undirected Networks}
\author{Gianlorenzo D'Angelo\\
Gran Sasso Science Institute \\
L'Aquila, Italy\\
gianlorenzo.dangelo@gssi.it\\
\And
Martin Olsen\\
Department of Business Development \\
 and Technology\\
 Aarhus University, Denmark\\
martino@btech.au.dk
\And
Lorenzo Severini\\
ISI Foundation, Turin, Italy \\
lorenzo.severini@isi.it
}

\begin{document}

\maketitle

\begin{abstract}
 Centrality metrics are among the main tools in social network analysis. Being central for a user of a network leads to several benefits to the user: central users are highly influential and play key roles within the network.
 Therefore, the optimization problem of increasing the centrality of a network user recently received considerable attention. Given a network and a target user $v$, the centrality maximization problem consists in creating $k$ new links incident to $v$ in such a way that the centrality of $v$ is maximized, according to some centrality metric.
 Most of the algorithms proposed in the literature are based on showing that a given centrality metric is monotone and submodular with respect to link addition. However, this property does not hold for several shortest-path based centrality metrics if the links are undirected.
 
 In this paper we study the centrality maximization problem in undirected networks for one of the most important shortest-path based centrality measures, the coverage centrality. We provide several hardness and approximation results. We first show that the problem cannot be approximated within a factor greater than $1-1/e$, unless $P=NP$, and, under the stronger gap-ETH hypothesis, the problem cannot be approximated within a factor better than $1/n^{o(1)}$, where $n$ is the number of users. We then propose two greedy approximation algorithms, and show that, by suitably combining them, we can guarantee an approximation factor of $\Omega(1/\sqrt{n})$. We experimentally compare the solutions provided by our approximation algorithm with optimal solutions computed by means of an exact IP formulation. We show that our algorithm produces solutions that are very close to the optimum. 
\end{abstract}

\section{Introduction}
\label{sec:introduction}
\noindent Determining what are the most important nodes in a network is one of the main goals of network analysis~\cite{N10}. Several so-called \emph{centrality metrics} have been proposed in the literature to try to quantitatively measure the importance of a node according to network properties like: distances between nodes (e.g. closeness or harmonic centrality), number of shortest paths passing through a node (e.g. betweenness or coverage centrality), or on spectral graph properties (e.g. PageRank or information centrality). 

It has been experimentally observed that nodes with high centrality values play key roles within the network. For example, closeness centrality is significantly correlated with the influence of users in a social network~\cite{CDSV16,MSHM12}, while shortest-path-based metrics are correlated with the number of passengers passing trough an airport in transportation networks~\cite{IETB12,MMPR09}.
The coverage centrality of a node $v$ is the number of distinct pairs of nodes for which a shortest path passes through $v$. Nodes with high coverage centrality are pivotal to the communication between many other nodes of the network.

Generally speaking, centrality metrics are positively correlated with many desirable properties of nodes, therefore, there has been a recent considerable interest on finding strategies to increase the centrality value of a given node in order to maximize the benefits for the node itself. In this paper we focus on the most used strategy which is that of modifying the network by adding a limited number of new edges incident to the node itself. In detail we study the following optimization problem: given a graph $G$, a node $v$ of $G$, and an integer $k$, find $k$ edges to be added incident to $v$ maximising the centrality value of $v$ in the graph $G$ augmented with these edges. The problem is usually referred to as the \emph{centrality maximization problem} and it can be instantiated by using different centrality metrics such as: PageRank~\cite{AN06,OV14}, eccentricity~\cite{DZ10,PBZ13}, coverage centrality~\cite{IETB12,MSSBS18}, betweenness~\cite{BCDMSV18,DSV15}, information centrality~\cite{SYZ18}, closeness and harmonic centrality~\cite{CDSV16}. The centrality maximization problem is in general $NP$-hard but in all the mentioned cases the authors were able to devise algorithms ensuring a constant approximation factor. In Table~\ref{tbl:bounds} we list the bounds on the approximation ratio reported in these references.

\begin{table}
{\renewcommand{\arraystretch}{1.2}
\begin{tabular}{@{}l|l|@{}c@{}|@{}c@{}}
Centrality& Graph & Approximation &Hardness of\\
metric     & type  & Algorithm & approximation\\
\hline
\hline
\multirow{2}{*}{Harmonic} & Undir. & $1-\frac{1}{e}$&$1-\frac{1}{15e}$  \\
\cline{2-4}
                          & Dir. & $1-\frac{1}{e}$&$1-\frac{1}{3e}$ \\
\hline
\multirow{2}{*}{Betweenness} & Undir. & OPEN&$1-\frac{1}{2e}$ \\
\cline{2-4}
                             & Dir. & $1-\frac{1}{e}$&$1-\frac{1}{2e}$ \\
\hline
Eccentricity                & undir. & $2+\frac{1}{OPT}$& $\frac{3}{2}$\\
\hline
PageRank                & Dir. & $\left(1-\alpha^2\right)\left(1-\frac{1}{e}\right)$&NO FPTAS \\
\hline
Information & Undir. & $1-\frac{1}{e}$&OPEN \\
\hline
Constr. coverage\!\! & Undir. & $1-\frac{1}{e}$&OPEN  \\
\hline
\multirow{3}{*}{Coverage} & \multirow{2}{*}{\textbf{Undir.}} & \multirow{2}{*}{$\mathbf{\Omega(1/\sqrt{n})}$}&$\mathbf{1-\frac{1}{e}}$  \\
&&&$\mathbf{1/n^{o(1)}}$\\
\cline{2-4}
                          & DAGs & $1-\frac{1}{e}$&OPEN
\end{tabular}
}
\caption{Summary of approximation bounds for the centrality maximization problem. The ``Constr. coverage'' row refers to the version of the coverage centrality maximization problem with the additional constraint that a pair of nodes can be covered by at most one edge. The results in this paper are marked in bold, the second hardness bound is under the Gap-ETH condition.}
\label{tbl:bounds}
\end{table}

Most of these approximation algorithms are based on a fundamental result on submodular optimization due to Nemhauser et al.~\cite{NWF78}. Given a monotone submodular set function $f$~\footnote{A set function $f$ is is submodular if for any pair of sets $A\subseteq B$ and any element $e\not\in B$, $f(A\cup\{e\})-f(A) \geq f(B\cup\{e\})-f(B)$.} and an integer $k$ the problem of finding a set $S$ with $|S|\leq k$ that maximises $f(S)$ is $NP$-hard and hard to approximate within a factor greater than $1-1/e$, unless $P=NP$~\cite{F98}. However, the greedy algorithm that starts with the empty set and repeatedly adds an element that gives the maximal marginal gain of $f$ guarantees the optimal approximation factor of $1-1/e$~\cite{NWF78}. Many of the approximation algorithms for the centrality maximization problem are based on the fact that the value for node $v$ of the considered centrality metric is monotone and submodular with respect to the addition of edges incident to $v$.

 \begin{figure}[t]
   \centering
  \includegraphics[scale=0.8]{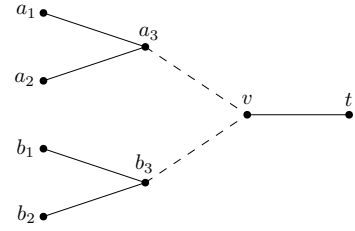}
\caption{Adding edge $\{a_3,v\}$ increases the coverage centrality of $v$ by 3, since nodes $a_i$ will have a shortest path to $t$ passing through $v$. Similarly, adding edge $\{b_3,v\}$ increases the coverage centrality of $v$ by 3. However, adding both edges $\{a_3,v\}$ and $\{b_3,v\}$ will increase the coverage centrality of $v$ by 15, since besides the shortest paths between nodes $a_i$ and $t$ and those between nodes $b_i$ and $t$ we need to take into account the 9 shortest paths between nodes $a_i$ and nodes $b_i$ that pass through $v$.}
\label{fig:example}
\end{figure}

Unfortunately, not all the centrality metrics exhibit a submodular trend. Indeed, it has been shown that in \emph{undirected graphs} some shortest-path based metrics are not submodular and, furthermore, the greedy algorithm exhibits an arbitrarily small approximation factor~\cite{DSV15}.
For example, Figure~\ref{fig:example} shows an undirected graph in which the increment in coverage centrality is not submodular with respect to edge addition.
This is in contrast with the results for the same centrality metrics on \emph{directed graphs}, where, e.g., betweenness and coverage centrality are monotone and submodular~\cite{BCDMSV18,IETB12}. 
Not being submodular makes things much harder and so far finding an approximation algorithm for the centrality maximization problem on shortest-paths based metrics has been left as an open problem~\cite{BCDMSV18,DSV15}.
To overcome this issue, Medya et al.~\cite{MSSBS18} consider the coverage centrality maximization problem with the additional artificial constraint that a pair of nodes can be covered by at most one edge. This constraint on the solution avoids the cases in Figure~\ref{fig:example} and makes the objective function submodular. However, it does not consider solutions that cover pairs of nodes with pairs of edges, hence it looks for sub-optimal solutions to the general problem.

In this paper we give the first results on the general coverage centrality maximization problem in undirected graphs. In the remainder of the paper we will refer to this problem as the \emph{Maximum Coverage Improvement} (\MCI) problem. Our results can be summarized as follows (see also Table~\ref{tbl:bounds}).
\begin{itemize}
 \item \MCI cannot be approximated within a factor greater that $1-1/e$, unless $P=NP$.
 \item \MCI is at least as hard to approximate as the well-known Densest-$k$-subgraph problem and hence cannot be approximated within any constant, if the Unique Games with Small Set Expansion conjecture~\cite{RS10} holds, and within $1/n^{o(1)}$, where $n$ is the number of nodes in the graph, if the Gap Exponential Time Hypothesis holds~\cite{M17}.
 \item We propose two greedy approximation algorithms for \MCI that guarantee, respectively, approximation factors of $1-e^{-\frac{(1-\epsilon)(t-1)}{k-1}}$ and $(1-\epsilon)(1-\frac{1}{e})^2\frac{k}{4n}$, where $t\geq 2$ is a constant tuning parameter and $\epsilon$ is any positive constant.
 \item We show that combining the two proposed algorithms we can achieve an approximation factor of $\Omega(1/\sqrt{n})$.
 \item We implemented the proposed algorithms and experimentally compared the solutions provided by our approximation algorithm with optimal solutions computed by means of an exact IP formulation. We experimentally show that our algorithm produces solutions that are very close to the optimum and that it is highly effective in increasing the coverage centrality of a target user.
\end{itemize}

\section{Notation and problem statement}
\label{sec:notation}

Let $G=(V,E)$ be an undirected graph where $|V|=n$ and $|E|=m$. 
For each node $v$, $N_v$ denotes the set of neighbors of $v$, i.e. $N_v=\{u~|~\{u,v\}\in E\}$.
Given two nodes $s$ and $t$, we denote by $d_{st}$ and $P_{st}$ the distance from $s$ to $t$ in $G$ and the set of nodes in any shortest path from $s$ to $t$ in $G$, respectively.
For each node $v$, the \emph{coverage centrality}~\cite{Y14} of $v$ is defined as the number of pairs $(s,t)$ such that $v$ is contained in a shortest path between $s$ and $t$, formally,
\begin{equation*}
 c_v = |\{ (s,t)\in V\times V ~|~ v\in P_{st}, v\neq s, v\neq t \}|.
\end{equation*}

In this paper, we consider graphs that are augmented by adding a set $S$ of arcs not in $E$. Given a set $S\subseteq (V\times V)\setminus E$ of arcs, we denote by $G(S)$ the graph augmented by adding the arcs in $S$ to $G$, i.e. $G(S)=(V,E\cup S)$. For a parameter $x$ of $G$, we denote by $x(S)$ the same parameter in graph $G(S)$, e.g. the distance from $s$ to $t$ in $G(S)$ is denoted as $d_{st}(S)$.

The coverage centrality of a node might change if the graph is augmented with a set of arcs. In particular, adding arcs incident to some node $v$ can increase the coverage centrality value of $v$. We are interested in finding a set $S$ of arcs incident to a particular node $v$ that maximizes $c_v(S)$. Therefore, we define the \emph{Maximum Coverage Improvement} (\MCI) problem as follows: Given an undirected graph $G=(V,E)$, a node $v\in V$, and an integer $k\in\mathbb{N}$, find a set $S$ of arcs incident to $v$, $S \subseteq\{\{u,v\}~|~u\in V\setminus N_v\}$, such that $|S|\leq k$ and $c_v(S)$ is maximized.

\section{Hardness of approximation}
We first show that \MCI cannot be approximated within a factor greater that $1-1/e$, unless $P=NP$. Then, we show that, under stronger conditions, it cannot be approximated to within a factor greater than $n^{-f(n)}$, for any $f\in o(1)$. 

\subsection{Constant bound}
Our first hardness of approximation result is obtained by reducing the \textit{Maximum Set Coverage} (\MSC) problem to \MCI. The problem \MSC is defined as follows: Given a ground set $U$, a collection ${F} = \{S_1,S_2,\ldots ,S_{|{F}|}\}$ of subsets of $U$, and an integer $k'\in\mathbb{N}$, find a sub-collection ${F}'\subseteq{F}$ such that $|{F}'|\leq k'$ and $s({F}') = |\cup_{S_i\in {F}'}S_i|$ is maximized.
It is known that the \MSC\ problem cannot be approximated within a factor greater than $1-\frac{1}{e}$, unless $P=NP$~\cite{F98}.

\begin{theorem}\label{th:reductionmc}
There is no polynomial time algorithm with approximation factor greater than $1-\frac{1}{e}$ for the \MCI problem on undirected graphs, unless $P=NP$.
\end{theorem} 
\begin{proof}


Assume that we have access to a polynomial time approximation algorithm $A_{\MCI}$ for the \MCI problem with approximation factor $1-\frac{1}{e}+\epsilon$ for some positive number $\epsilon$. We consider an instance $I_{\MSC}$ of the \MSC problem and build the instance $I_{\MCI}$ of \MCI shown in Fig.~\ref{fig:reduction_MSC}.     
\begin{figure}
  \centering 
  \includegraphics[scale=.57]{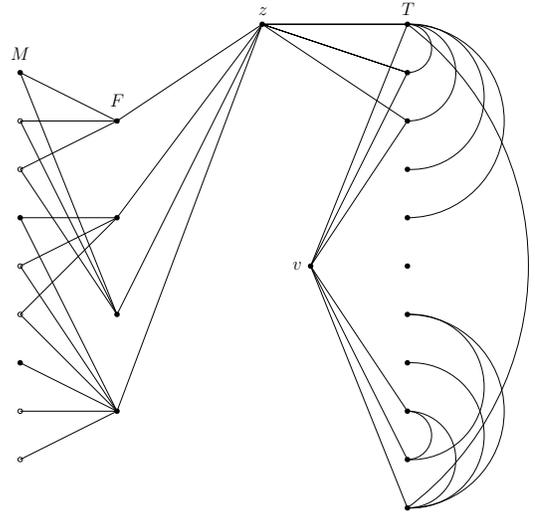}
  \caption{The $I_{\MCI}$ instance used in Theorem~\ref{th:reductionmc}.}
  \label{fig:reduction_MSC}
\end{figure}
The instance consists of $5$ vertical layers of nodes. For each element of $U$ of the \MSC instance, we have a member (filled circle) and $D$ copies (unfilled circles) in the $M$-layer to the left. In the $F$-layer there is a node for each set in the collection $F$ in the \MSC instance. A node in the $F$-layer is connected to all the corresponding members and copies in the $M$-layer. In the third layer there is a single node $z$ connecting all nodes in the $F$-layer to all nodes in the $T$-layer to the far right. In the fourth layer we have the node $v$ that is connected to all nodes in the $T$-layer. All the nodes in the $T$-layer to the right form a clique. Note that not all the edges are shown in the figure.
Let $\beta > 0$ be a sufficiently small positive constant satisfying
$$1-\frac{1}{e} < \frac{1}{1+\beta}\left(1-\frac{1}{e}+\epsilon\right) \enspace .$$
To improve the readability, we will not distinguish between a set and its cardinality. E.g., $T$ can represent the set of all the $T$-nodes and $T$ can also represent the number of $T$-nodes. Our aim is to choose relatively small $D$ and $T$ such that
\begin{equation}\label{eq:CT}
(F+M+1)^2 \leq (\beta (D+1)-k-F)T \enspace .
\end{equation}
We choose $D$ and $T$ as follows (note that $M = U(D+1)$):
$$D = \left\lceil\frac{1+k+F}{\beta}-1\right\rceil \quad \quad T = (F+M+1)^2\enspace .$$
Since $\beta$ is a constant,  $T$ and $M$ are polynomial in $|I_{\MSC}|$.

Now let $S_{\MCI}$ be the solution computed by $A_{\MCI}$ given the $I_{\MCI}$ instance as input. Let $S_{\MSC}$ be the solution for \MSC corresponding to all the sets for which there is an edge between the $F$-node and $v$ in $S_{\MCI}$. Note that $S_{\MSC}$ can be computed in polynomial time.
Let $A$, $A_1$ and $A_2$ be defined as follows:
$$A = \{(s, t) \in V \times V~|~v \in P_{st}(S_{\MCI})\}$$
$$A_1 = \{(s, t) \in A~|~(s, t) \in (M \times T) \cup (T \times M)\}$$
$$A_2 = A \setminus A_1 \enspace .$$
We now have the following identity $c_v(S_{\MCI}) = A_1+A_2$.
The set $A_1$ consists of all the pairs of vertices with one element in $M$ and one element in $T$ that are covered by $v$ in the graph corresponding to $S_{\MCI}$. The contribution to $A_1$ of the edges in $S_{\MCI}$ with one element in $F$ is $2(D+1)T \cdot s(S_{\MSC})$ and the contribution of edges in $S_{\MCI}$ with one element in $M$ is no more than $2kT$ and there might be some overlap. This allows us to establish the following upper bound on $A_1$:
\begin{equation*}
\label{eq:A1}
A_1 \leq 2(D+1)T \cdot s(S_{\MSC}) + 2kT \enspace .
\end{equation*}
The remaining pairs that might be covered 1) have an element in $F$ and an element in $T$, or 2) have no elements in $T$: $A_2 \leq 2TF+2(F+M+1)^2 \enspace .$

According to (\ref{eq:CT}), $D$ and $T$ have been chosen such that
$kT+TF+(F+M+1)^2 \leq \beta (D+1)T$,
implying
$c_v(S_{\MCI})\leq 2(1+\beta)(D+1)T \cdot s(S_{\MSC})$.
If we add edges to $v$ in the \MCI instance corresponding to the optimal solution of the \MSC instance, we obtain a feasible solution for the \MCI instance. For each covered element in the \MSC instance, we have $2(D+1)T$ covered pairs in the \MCI instance, therefore 
$2(D+1)T \cdot OPT(I_{\MSC}) \leq OPT(I_{\MCI})$.

The algorithm $A_{\MCI}$ has approximation factor $1-\frac{1}{e}+\epsilon$:
$$\frac{c_v(S_{\MCI})}{OPT(I_{\MCI})} \geq 1-\frac{1}{e}+\epsilon \enspace .$$
This allows us to set up the following inequality:
$$\frac{2(1+\beta)(D+1)T \cdot s(S_{\MSC})}{2(D+1)T \cdot OPT(I_{\MSC})} \geq 1-\frac{1}{e}+\epsilon \enspace .$$
We can now establish a lower bound for the approximation factor for the solution to our \MSC instance:
$$\frac{s(S_{\MSC})}{OPT(I_{\MSC})} \geq \frac{1}{1+\beta}\left(1-\frac{1}{e}+\epsilon\right) > 1-\frac{1}{e}\enspace ,$$
a contradiction.
\end{proof}

\subsection{Conditional bound}
To obtain our next hardness result, we reduce the \emph{Densest-$k$-Subgraph} (\DkS) problem to \MCI.  \DkS is defined as follows: Given a graph $G$ and an integer $k$, find a subgraph of $G$ induced on $k$ vertices with the maximum number of edges.


Several conditional hardness of approximation results for \DkS have been proved (see e.g.~\cite{M17} and references therein).
It has been shown that \DkS is hard to approximate within any constant bound under the Unique Games with Small Set Expansion conjecture~\cite{RS10}. Recently, it has been shown that under the exponential time hypothesis (ETH) there is no polynomial-time algorithm that approximates \DkS to within $n^{-1/(loglogn)^c}$, for some constant $c$. Moreover, under the stronger Gap-ETH assumption, the factor can be improved to $n^{-f(n)}$ for any function $f\in o(1)$~\cite{M17}.
The next theorem shows that there is an S-reduction~\cite{Crescenzi1997} from \DkS to \MCI only adding a constant to the number of nodes. Then, all the above mentioned inapproximability results extend to the \MCI problem. The current state-of-the-art algorithm for \DkS guarantees a  $\Omega(n^{-\frac{1}{4}-\epsilon})$ approximation~\cite{BhaskaraCCFV10}.

\begin{theorem}
There is an S-reduction from \DkS to \MCI. The reduction transforms a \DkS instance with $n$ vertices into an \MCI instance with $n+2$ vertices. 
\end{theorem}

\begin{proof} 
Consider a \DkS instance given by the graph $G'(V', E')$ and an integer $k'$. This instance is transformed into an \MCI instance given by the graph $G(V,E)$, a node $v \in V$ and an integer $k$ as follows. The set of nodes $V$ is formed by adding a node $x$ and the target node $v$ for the \MCI instance to $V'$: $V = \{x, v\} \cup V'$. The node $v$ is isolated and $x$ is connected to all other nodes in $V$. In addition to the edges incident to $x$, the edges in the complement of the original graph $G'$ are also added to the graph. Formally, $E = \{\{u,w\}~|~ \{u,w\}\not\in E'\}\cup \{\{x,u\}~|~u\in V' \cup\{v\}\}$. The value of $k$ is not changed: $k=k'$. See Fig.~\ref{fig:reduction_DkS}. We also need to explain how to transform a feasible solution of the \MCI instance into a feasible solution of the \DkS instance. Here we simply pick the nodes that are linking to $v$ in the feasible solution for the \MCI instance (excluding $x$). 

We now  prove that the reduction is an S-reduction. We claim that the following holds: The node $v$ is on the shortest path between $s$ and $t$ in the \MCI instance after adding $k$ edges if and only if 1) edges $\{v, s\}$ and $\{v, t\}$ are added, and 2) there is an edge between $s$ and $t$ in $G'$. The if-direction is clear. To prove the only-if-direction, we assume that $v$ is on a shortest path between the nodes $s$ and $t$ in the \MCI instance. If 1) is false, then the length of the shortest path through $v$ is at least $3$. If 2) is false, we also arrive at a contradiction.

This implies that the number of edges induced by a feasible solution of the \DkS instance is precisely the coverage centrality of $v$ in the corresponding \MCI instance and vice versa. This shows that the reduction is an S-reduction.
\end{proof}

\begin{figure}
  \centering 
  \includegraphics[scale=0.7]{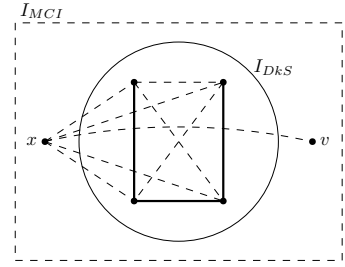}
  \caption{The reduction from \DkS. In the circle, we have the original \DkS instance (solid edges) that is transformed into the \MCI instance indicated by the rectangle (dashed edges).}
  \label{fig:reduction_DkS}
\end{figure}

\section{Approximation algorithm}

It is easy to see that, in the undirected case, the objective function is not submodular and that there are instances of \MCI (similar to that in Figure~\ref{fig:example}) for which the greedy algorithm by Nemhauser et al. exhibits an arbitrarily small approximation factor.
The main problem with the greedy algorithm is that it does not take into account the shortest paths that pass through two of the added edges. In this section, we show how to overcome this limitation and we give an algorithm that guarantees a $\Omega(1/\sqrt{n})$-approximation. 

The algorithm is based on a reduction to a generalization of the maximum coverage problem in which elements of a ground set are covered by pairs of ``objects'', instead of a single set, and we look for a bounded set of objects that maximizes the number of covered elements. We call this problem the \emph{Maximum Coverage with Pairs problem} (\MCPP). Formally, \MCPP is defined as: Given a ground set $X$, a set $O$ of objects, and an integer $k\in\mathbb{N}$, find a set $O'\subseteq O$, such that $|O'|\leq k$ and $c(O') = |\cup_{i,j\in O'}C(i,j)|$ is maximum, where $C(i,j)$ denotes the subset of $X$ covered by pair $\{i,j\}$, for each unordered pair of objects $\{i,j\}$.


Given $O'\subseteq O$, let $C(O') = \cup_{i,j\in O'}C(i,j)$ and $c(O')=|C(O')|$. Wlog, we assume that each element in $X$ is covered by at least a pair of objects in $O$ and that $k\leq |O|$. 

The problem \MCI can be reduced to \MCPP as follows:
 for each pair $(s,t)$ of nodes in $G$, we add an element $(s,t)$ to $X$;
 for each $u\in V\setminus N(v)$, we add an object $i_u=\{u,v\}$ to $O$ (i.e. all the edges that can be added incident to $v$);
 for each pair of objects $i_u,i_w\in O$, we set $C(i_u,i_w) = \{ (s,t)~|~ v\in P_{st}(\{\{u,v\},\{w,v\}\})\}$;
 we set $k'=k$.
%
Any feasible solution $O'$ to the above instance of \MCPP corresponds to a feasible solution $S=O'$ for \MCI. Since for each pair of nodes $(s,t)$ in $V$ the shortest path between $s$ and $t$ in $G(S)$ can only pass through at most two edges of $S$, then $c_v (S) = |\cup_{i_u,i_w\in O'}C(i_u,i_w)| = c(O')$. Therefore, any approximation algorithm for \MCPP can be used to solve \MCI with the same approximation factor.

We observe that \MCPP is a generalization of the \DkS problem,
which corresponds to the case in which $|C(i,j)|\leq 1$, for $i,j\in O$, and each element of $X$ is covered by exactly one pair of objects (i.e. objects correspond to nodes and elements correspond to edges). Therefore, \MCPP is at least as hard to approximate as \DkS.

Our algorithm exploits two procedures, called \Greedy1 and \Greedy2, that return two sets of objects, and selects one of these sets that gives the maximum coverage. In particular, Procedure \Greedy1 returns a set that guarantees an approximation factor of $\left(1-e^{-\frac{(1-\epsilon)(t-1)}{k-1}}\right)$, where $t\in [2,k]$ is a constant integer parameter of the procedure and $\epsilon$ is any positive constant, while Procedure \Greedy2 guarantees an approximation factor of $(1-\epsilon)\frac{1}{4}\left(1-\frac{1}{e}\right)^2\frac{k}{|O|}$ for any constant $\epsilon >0$ (see Theorems~\ref{th:greedyone} and~\ref{th:greedytwo}). The next theorem shows the overall approximation factor. When applied to the \MCI problem, it guarantees a $\Omega(1/\sqrt{n})$-approximation.

\begin{theorem}
 Let $O^*$ be an optimum solution for \MCPP, let $O_1$ and $O_2$ be the solutions of Procedures \Greedy$1$ and \Greedy$2$, then
 \[
   \max\{c(O_1),c(O_2)\}\geq (1-\epsilon)\frac{1}{2}\left(1-\frac{1}{e}\right)^{3/2}\sqrt{\frac{t-1}{|O|}}c(O^*), 
 \]
for any constant $t\geq 2$ and $\epsilon \in (0,1)$.
\end{theorem}
\begin{proof}
 The value of $\max\{c(O_1),c(O_2)\}$ is at least the geometric mean of $c(O_1)$ and $c(O_2)$. Moreover, $\left(1-e^{-\frac{(1-\epsilon)(t-1)}{k-1}}\right)\geq\left(1-\frac{1}{e}\right)\frac{(1-\epsilon)(t-1)}{k-1} $, for any $\epsilon\in(0,1)$ and $k> 1$\footnote{Indeed, $1-e^{-x}\geq (1-e^{-1})x$, for any $x\in [0,1]$, and $\frac{(1-\epsilon)(t-1)}{k-1}\in [0,1]$ for any $\epsilon\in(0,1)$, $t\leq k$, and $k> 1$.}. Therefore,
 \begin{align*}
    &\max\{c(O_1),c(O_2)\}\geq \sqrt{c(O_1)\cdot c(O_2)} \\
                         &\geq\! \sqrt{\! \left(\!1-\frac{1}{e}\!\right)\!\!\frac{(1-\epsilon)(t-1)}{k-1} c(O^*)\!\cdot\!(1-\epsilon)\frac{1}{4}\!\left(\!1-\frac{1}{e}\!\right)^{\!\!2}\!\frac{k}{|O|}c(O^*) }\\
                         &\geq (1-\epsilon)\frac{1}{2}\left(1-\frac{1}{e}\right)^{3/2}\sqrt{\frac{t-1}{|O|}}c(O^*).\qedhere
 \end{align*}
\end{proof}

We now introduce procedures \Greedy1 and \Greedy2.

\subsection{Procedure \Greedy1}

\begin{algorithm2e}[t]
  $O':=\emptyset$\;
  \While{$|O'|\leq k - t$}
  {
    $Z:=\arg\max_{Z\subseteq O , |Z|\leq t}\{C(O' \cup Z) - C(O')\}$\;
%
    $O':=O'\cup Z$\;
  }
  
   $Z:=\arg\max_{Z\subseteq O , |Z|\leq k-|O'|}\{C(O' \cup Z) - C(O')\}$\;\label{algo:undirectedpairs:finalone}
  $O':=O'\cup Z$\;\label{algo:undirectedpairs:finaltwo}
  \caption{Procedure \Greedy1}
  \label{algo:undirectedpairs}
\end{algorithm2e}

The pseudo-code of Procedure \Greedy1 is reported in Algorithm~\ref{algo:undirectedpairs}.
For some fixed constant integer $t\in[2,k]$, the procedure greedily selects a set of objects of size $t$ that maximizes the increment in the objective function. In particular, it starts with an empty solution and iteratively adds to it a set $Z$ of $t$ objects that maximizes $c(O' \cup Z) - c(O')$, where $O'$ is the solution computed so far. The procedure stops when it has added at least $k-t$ objects to $S$. Eventually, if $|O'| < k$, it completes the solution by adding a further set of $k-|O'| < t$ objects (lines~\ref{algo:undirectedpairs:finalone}--\ref{algo:undirectedpairs:finaltwo}).
Note that one or more objects of the selected set might already belong to $O'$ (but not all of them). Hence, Algorithm~\ref{algo:undirectedpairs} has at least $\left\lfloor\frac{k-t}{t}\right\rfloor$ and at most $k$ iterations.

For each iteration $i$ of Algorithm~\ref{algo:undirectedpairs}, let $O_i$ be the set $O'$ at the end of iteration $i$ and let $O^*$ be an optimal solution. The next lemma is used to prove the approximation bound, the full proof can be found in the supplementary material.
\begin{lemma}\label{lem:undirectedpairs:two}
  After each iteration $i$ of Algorithm~\ref{algo:undirectedpairs}, the following holds 
  \begin{equation}\label{eq:undirectedpairs:two}
   c(O_i) \geq \left(1-\left(1-\frac{t(t-1)}{k(k-1)}\right)^i \right) c(O^*).
  \end{equation}
\end{lemma}

\begin{theorem}\label{th:greedyone}
If $I$ is the number of iterations of Algorithm~\ref{algo:undirectedpairs}, then
 \[
  c(O_I) \geq \left(1-e^{-(1-\epsilon)\frac{t-1}{k-1}}\right) c(O^*)\enspace ,
 \]
 for any constant $\epsilon\in (0,1)$.
\end{theorem}
\begin{proof}
We observe that for any constant $\epsilon>0$, there exists a constant $k_0$, such that for each $k\geq k_0$, $I = \left\lfloor\frac{k-t}{t}\right\rfloor\geq \frac{k}{t} - 2\geq (1-\epsilon)\frac{k}{t}$. Note that, when $k$ is a constant the problem can be easily solvable in polynomial time by brute force and therefore we assume that $k$ is not a constant. 
%
Plugging $I$ into inequality~\eqref{eq:undirectedpairs:two}, we obtain:
 \begin{align*}
   c(O_I) &\geq \left(1-\left(1-\frac{t(t-1)}{k(k-1)}\right)^{ \left\lfloor\frac{k-t}{t}\right\rfloor} \right) c(O^*).\\
              & \geq \left(1-\left(1-\frac{t(t-1)}{k(k-1)}\right)^{ (1-\epsilon)\frac{k}{t}} \right) c(O^*).
 \end{align*}
 By calculus, it can be shown that $1-x \leq e^{-x}$, which implies that $\left(1-\frac{t(t-1)}{k(k-1)}\right)^{(1-\epsilon)\frac{k}{t}} \leq e^{-(1-\epsilon)\frac{t-1}{k-1}}$, and finally:
 \[
  \!\left(\!\!1\!-\!\left(\!\!1\!-\!\frac{t(t-1)}{k(k-1)}\!\right)^{\!(1-\epsilon)\frac{k}{t}} \right)\! c(O^*) 
 \!\geq\! \left(\!1\!-\! e^{-(1-\epsilon)\frac{t-1}{k-1}}\!  \right)\! c(O^*). \qedhere
 \]
\end{proof}

\subsection{Procedure \Greedy2}

\begin{algorithm2e}[t]

Define an instance of \MSC made of ground set $X$ and, for each object $o\in O$, a set equal to $N(o)$\;
Run the greedy algorithm in~\cite{NWF78} for \MSC to find a set $H$ of size $\left\lceil\frac{k}{2}\right\rceil$ objects\label{algo:FKP2:h}\;
Define an instance of \MSC made of ground set $D(H)=N(H)\setminus C(H)$ and, for each object $o\in O\setminus H$, a set equal to $\cup_{i \in H} C(o,i) \setminus C(H)$\;
Run the greedy algorithm in~\cite{NWF78} for \MSC to find a set $I$ of size $\left\lfloor\frac{k}{2}\right\rfloor$ objects\;
\Return $H\cup I$\;
  \caption{Procedure \Greedy2}
  \label{algo:FKP2}
\end{algorithm2e}

In order to describe Procedure \Greedy2, we need to introduce further notation.
Given a set $O'$ of objects, we denote by $N(O')$ the set of elements that the objects in $O'$ can cover when associated with any other object in $O$, that is $N(O') = \cup_{o\in O',i\in O}C(o,i)$. The \emph{degree} of $O'$ is the cardinality of $N(O')$ and it is denoted by $d(O')$. To simplify the notation, when $O'$ is a singleton, $O'=\{o\}$, we use  $N(o)$ and $d(o)$ to denote $N(O)$ and $d(O)$, respectively. Intuitively, $N(O')$ are the elements that are covered by at least an object in $O'$, while $C(O')$ are the elements that are covered by at least two objects in $O'$. In the following, we say that an element is \emph{single} covered by $O'$ in the former case and \emph{double} covered by $O'$ in the latter case. We observe that $d(O')\geq c(O')$, for any set of objects $O'$. Moreover, since the degree is defined as the size of the union of sets, then it is a monotone and submodular set function.

Procedure \Greedy2 is given in Algorithm~\ref{algo:FKP2}. First, the procedure looks for a set $H$ of $\left\lceil\frac{k}{2}\right\rceil$ objects with maximum degree. Since computing such a set is equivalent to solving an instance of \MSC, which is known to be $NP$-hard, we compute an approximation of it. In detail, the instance of \MSC is made of the same ground set $X$ and, for each object $o\in O$, a set equal to $N(o)$. Any set of objects $O'$ corresponds to a solution to this \MSC instance, where the number of single covered elements is equal to $d(O')$. Indeed, finding a set of $\left\lceil\frac{k}{2}\right\rceil$ objects that maximizes the degree in the \MCPP instance corresponds to finding a collection of sets that maximizes the single coverage of $X$ in this \MSC instance. Hence, we find a set $H$ that approximates the maximum single coverage of $X$, and, in particular, we exploit a greedy algorithm that guarantees an optimal approximation of $1-1/e$~\cite{NWF78}. 
Then, the procedure selects a set of $\left\lfloor\frac{k}{2}\right\rfloor$ objects in $O\setminus H$ that maximizes the single coverage of the elements in $N(H)$ not double covered by $H$. In other words, these objects, {\em along with} objects in $H$, double cover the maximum fraction of $N(H)$.
Again, computing such a set is equivalent to solving an instance of \MSC, and we find a $(1-1/e)$-approximation. In this case, the \MSC instance is made of the ground set $D(H) = N(H) \setminus C(H)$, and for each object $o\in O\setminus H$, a set equal to $\cup_{i \in H} C(o,i) \setminus C(H)$. The approximated solution found by the greedy algorithm for \MSC is denoted by $I$.
Procedure \Greedy2 outputs the set of objects $H\cup I$.

In the next lemma we establish a connection between the number of single and double covered elements, in particular, we show an upper bound to the optimal value $c(O^*)$ of the \MCPP instance as a function of $d(H)$, where $H$ is the set of objects selected at line~\ref{algo:FKP2:h} of Algorithm~\ref{algo:FKP2}. The full proof of the lemma can be found in the supplementary material.
\begin{lemma}\label{lem:FKP}
If $H$ is the set of objects selected at line~\ref{algo:FKP2:h} of Algorithm~\ref{algo:FKP2} and $O^*$ is an optimal solution for \MCPP, then
 \[
  d(H) \geq \frac{1}{2}\left(1-\frac{1}{e}\right)c(O^*).
 \]
\end{lemma}

\begin{theorem}\label{th:greedytwo}
If $H\cup I$ is the output of Algorithm~\ref{algo:FKP2} and $O^*$ is an optimal solution for \MCPP, then
\[
  c(H\cup I) \geq (1-\epsilon)\frac{1}{4}\left(1-\frac{1}{e}\right)^2\frac{k}{|O|} c(O^*),
\]
for any constant $\epsilon \in (0,1)$.
\end{theorem}
\begin{proof}
The objects in $H$ double cover $c(H)$ elements, hence, the number of elements that are single coreved by $H$ but not double covered by $H$ is $d(H) - c(H)$. The set of these elements is $D(H) = N(H) \setminus C(H)$.
We now show that $N(I)$ contains at least a fraction $\left(1-\frac{1}{e}\right)\left\lfloor\frac{k}{2}\right\rfloor\frac{1}{|O|}$ of these elements and hence the objects in $H\cup I$ double cover at least $\left(d(H) - c(H)\right)\left(1-\frac{1}{e}\right)\left\lfloor\frac{k}{2}\right\rfloor\frac{1}{|O|}$ of them.

Let us denote as $I^*$ a set of $\left\lfloor\frac{k}{2}\right\rfloor$ objects in $O\setminus H$ that maximizes the single coverage of elements in $D(H)$, that is the size of $D(H) \cap N(I^*)$ is maximum for sets of $\left\lfloor\frac{k}{2}\right\rfloor$ objects.
By contradiction, let us assume that the size of $D(H) \cap N(I^*)$ is smaller than $\left(d(H) - c(H)\right)\left\lfloor\frac{k}{2}\right\rfloor\frac{1}{|O|}$. 

Let us partition $O\setminus H$ into sets of $\left\lfloor\frac{k}{2}\right\rfloor$ objects plus a possible set of smaller size, if $|O\setminus H|$ is not divisible by $\left\lfloor\frac{k}{2}\right\rfloor$. The number of the sets in the partition is 
\[
\ell= \left\lceil \frac{|O\setminus H|}{\left\lfloor\frac{k}{2}\right\rfloor} \right\rceil = \left\lceil \frac{|O|-\left\lceil\frac{k}{2}\right\rceil}{\left\lfloor\frac{k}{2}\right\rfloor} \right\rceil \leq \frac{|O|}{\left\lfloor\frac{k}{2}\right\rfloor}.
\]
We denote the sets of this partition as $I_i$, $i=1,2,\ldots, \ell$. Since $I^*$ maximizes the single coverage of elements in  $D(H)$, then for each $I_i$, the size of $D(H) \cap N(I_i)$ is smaller than $\left(d(H) - c(H)\right)\left\lfloor\frac{k}{2}\right\rfloor\frac{1}{|O|}$. By submodularity we have
\begin{align*}
 &|D(H) \cap N(O\setminus H)| \leq \sum_{i=1}^\ell |D(H) \cap N(I_i)| \\
 &< \sum_{i=1}^\ell \left(d(H) - c(H)\right)\left\lfloor\frac{k}{2}\right\rfloor\frac{1}{|O|}\leq d(H) - c(H), 
\end{align*}
which is a contradiction because it implies that there are elements in $D(H)$ that are not covered by any object in $O\setminus H$ and hence they cannot be double covered. This proves that  the size of $D(H) \cap N(I^*)$ is at least $\left(d(H) - c(H)\right)\left\lfloor\frac{k}{2}\right\rfloor\frac{1}{|O|}$. Moreover, set $I$ approximates the optimal single coverage of $D(H)$ by a factor $1-\frac{1}{e}$ and hence the size of  $D(H)\cap N(I)$ is at least $\left(1-\frac{1}{e}\right)\left(d(H) - c(H)\right)\left\lfloor\frac{k}{2}\right\rfloor\frac{1}{|O|}$.

It follows that the overall number of elements double covered by $H\cup I$ is at least 
$
 \left(d(H) - c(H)\right)\left(1-\frac{1}{e}\right)\left\lfloor\frac{k}{2}\right\rfloor\frac{1}{|O|} + c(H).
$
By Lemma~\ref{lem:FKP}, this values is at least

\[
 \left(\frac{1}{2}\left(1-\frac{1}{e}\right)c(O^*)- c(H)\right)\left(1-\frac{1}{e}\right)\left\lfloor\frac{k}{2}\right\rfloor\frac{1}{|O|} + c(H).
\]
For any constant $\epsilon>0$ and $k$ greater than a constant, $\left\lfloor\frac{k}{2}\right\rfloor\geq (1-\epsilon)\frac{k}{2}$, then this number is at least 
\[
 (1-\epsilon)\left(\frac{1}{2}\left(1-\frac{1}{e}\right)c(O^*)- c(H)\right)\left(1-\frac{1}{e}\right)\frac{k}{2|O|} + c(H)
\]
\begin{align*}
  &= (1-\epsilon)\frac{k}{4|O|}\left(1-\frac{1}{e}\right)^2c(O^*) \\
  &\hspace{2.5cm}+ c(H)\left( 1 - (1-\epsilon)\left(1-\frac{1}{e}\right)\frac{k}{2|O|}\right)
\end{align*}
\[
 \geq (1-\epsilon)\frac{k}{4|O|}\left(1-\frac{1}{e}\right)^2c(O^*), 
\]
since $c(H)\geq 0$ and $k\leq 2|O|$.
\end{proof}

\section{Experimental study}
In this section, we study the algorithms \Greedy1 and \Greedy2 from an experimental point of view. First, we compare the solutions of the greedy algorithms with optimal solutions computed by using an integer program formulation of \MCPP\ in order to assess the real performance in terms of solution quality (see the Supplementary material for the detailed implementation of the IP formulation).
Then, we focus on the \MCI problem and compare \Greedy1 and  \Greedy2 with the natural algorithm that adds $k$ random edges. 
We execute our experiments on two popular model networks, the Barabasi-Albert (BA) network~\cite{BA99} and the Configuration Model
(CM) network~\cite{BC78,MR95}, and on real-world networks extracted from human activities\footnote{\url{http://konect.uni-koblenz.de/}}. The sizes of the networks are reported in Table~\ref{tab:networks}.
All our experiments have been performed on a computer equipped with an Intel Xeon E5-2643 CPU clocked at 3.4GHz and 128GB of main memory, and our programs have been implemented in \texttt{C++}.

\begin{table}
	\centering
	\begin{tabular}{ccc}
		\hline
		\textbf{Network} & $n=|V|$ & $m=|E|$\\
		\hline
		\texttt{BA} & 50 & 96 \\
		\texttt{CM} & 50 &  85\\
		\texttt{karate} &34  &78 \\
		\texttt{windsurfers} &43  &336  \\
		\texttt{jazz}&198&2742\\
		\texttt{haggle}&274&2899\\
		\hline
	\end{tabular}
	\caption{Undirected networks used in the experiments.}
	\label{tab:networks}
\end{table}


The results of the comparison with the optimum are reported in Figure~\ref{Fig:opt}. For each network, we randomly choose 10 target nodes and, for each target node $v$, we add $k$ 
nonexistent edges incident to $v$ for $k=1,2,\ldots,10$.
Then, we plot the average coverage centrality of the 10 target nodes for each $k$.  We observe that there is little difference between the solutions of \Greedy1  and \Greedy2 and the optimal solutions, since the approximation ratio of \Greedy1 is always greater than 0.97 while the approximation ratio of \Greedy2 is always greater than 0.78.

\begin{figure}[t]
\centering
	\begin{minipage}{0.23\textwidth}
		\centering
		\includegraphics[width=.95\linewidth]{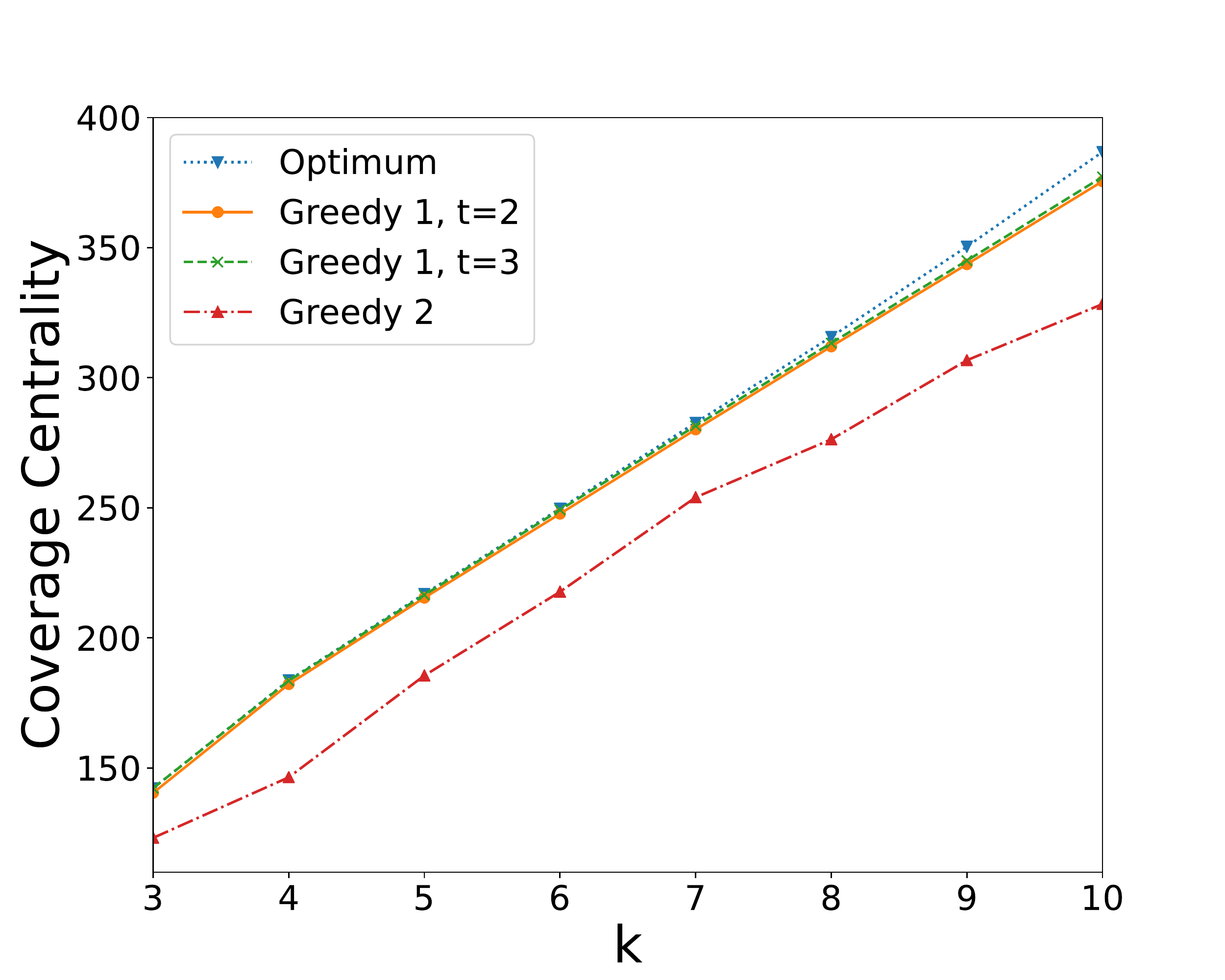}
		{\texttt{BA}}\label{Fig:opt:ba}
	\end{minipage}
	\begin{minipage}{0.23\textwidth}
		\centering
		\includegraphics[width=.95\linewidth]{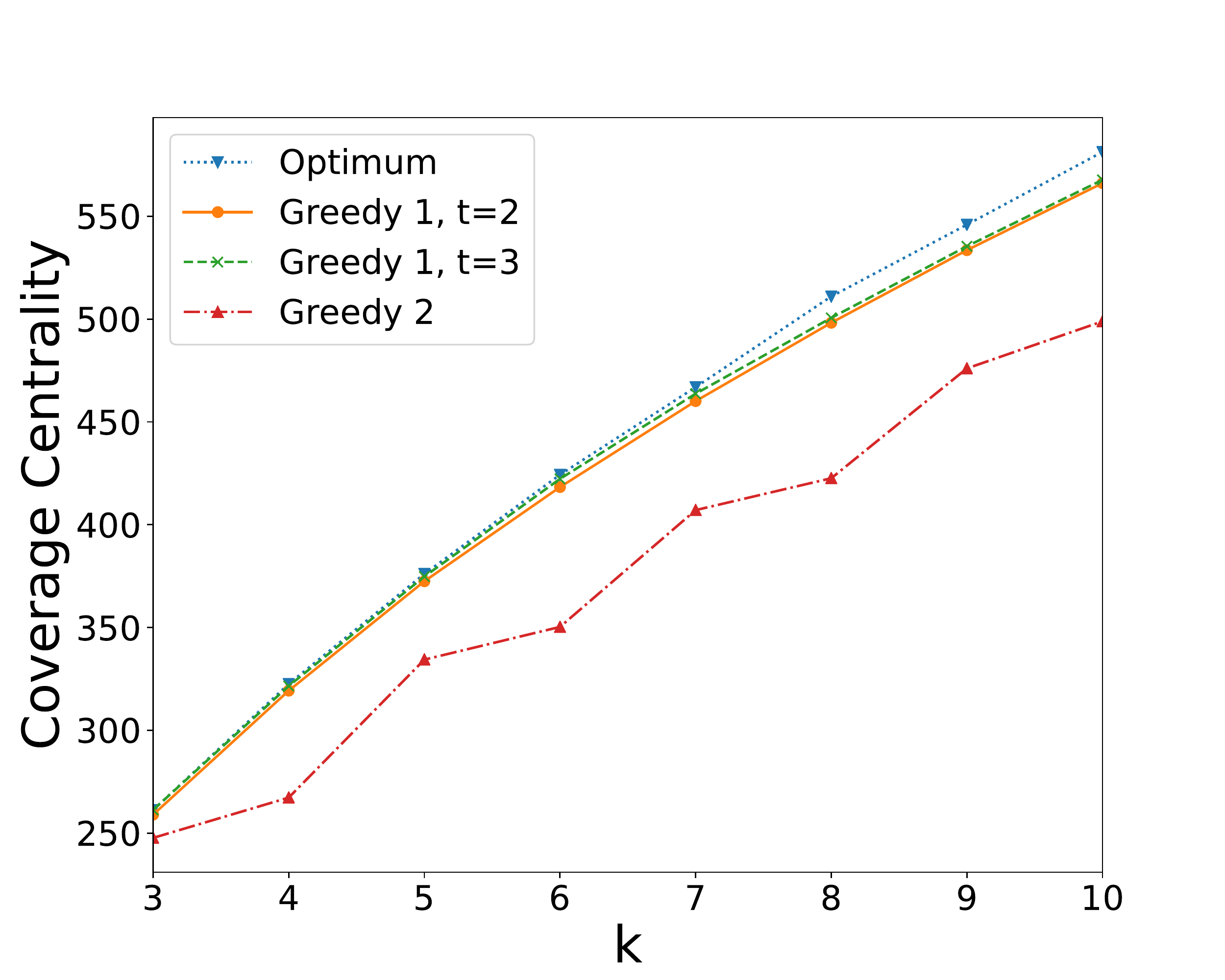}
		{\texttt{CM}}\label{Fig:opt:cm}
	\end{minipage}
\\
\begin{minipage}{0.23\textwidth}
		\centering
		\includegraphics[width=.95\linewidth]{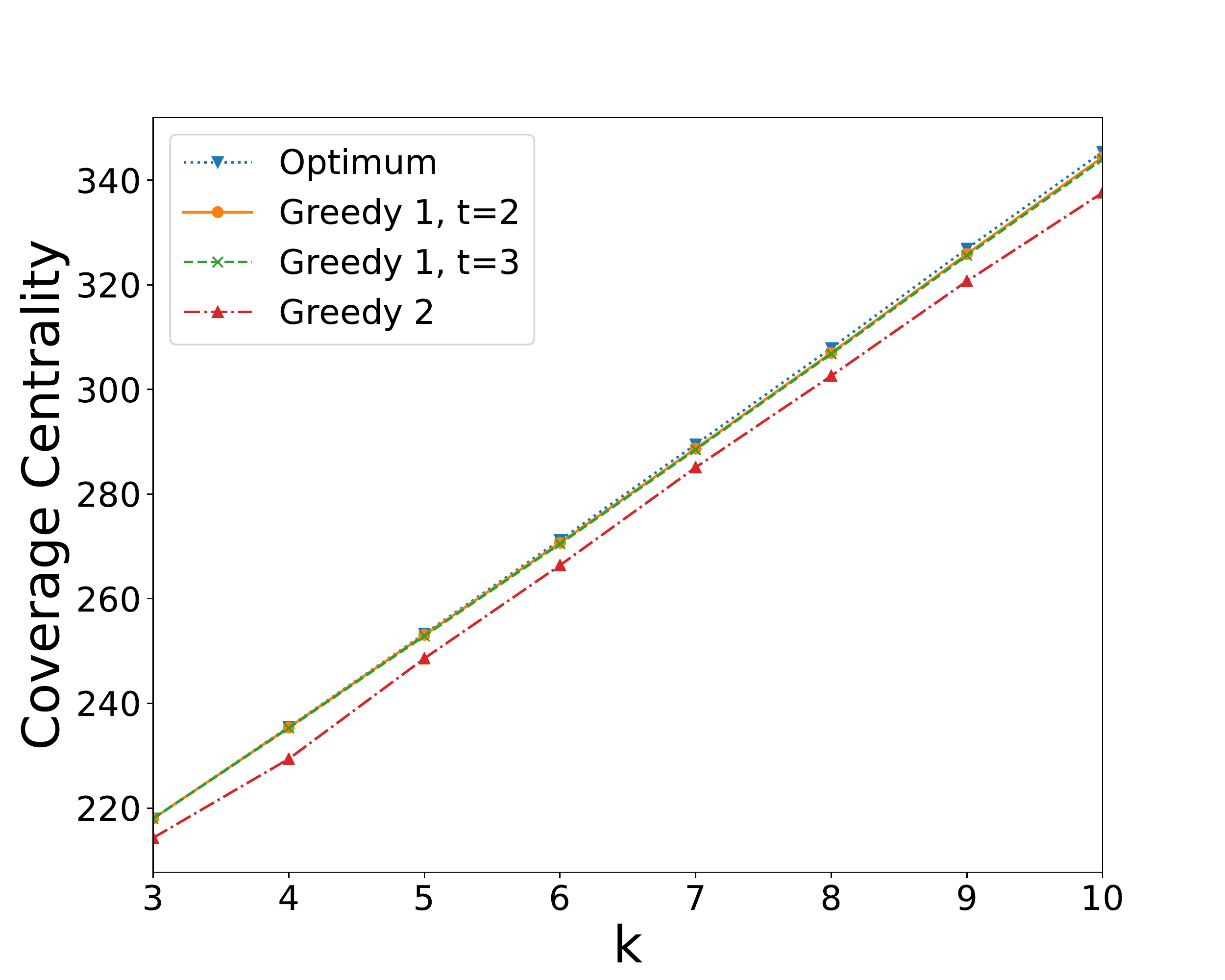}
		{\texttt{karate}}\label{Fig:opt:karate}
	\end{minipage}
	\begin{minipage}{0.23\textwidth}
		\centering
		\includegraphics[width=.95\linewidth]{{plot/ba50m2t1.csv.opt}.pdf}
		{\texttt{windsurfers}}\label{Fig:opt:windsurfers}
	\end{minipage}
	\caption{Average coverage centrality of target nodes as a function of the number $k$
		of inserted edges for \Greedy1 (with $t=2,3$), \Greedy2, and optimal solutions.}
	\label{Fig:opt}
\end{figure}
%
%
%


It is not possible to find the optimum on networks with thousand of edges in a reasonable time. Therefore, we compare the solutions with the natural baseline of adding $k$ random edges incident to each target node (the \Random algorithm). Analogously to the previous case, we plot the average coverage centrality of the 10 target nodes for each $k$. The results are reported in Figure~\ref{Fig:random:value}. We notice that also in this case \Greedy1 provides a better solution than \Greedy2. However, both algorithms perform always better than the \Random algorithm. 
On \texttt{jazz}, \Greedy1 with $t=2$ needs 7.5 seconds to solve the problem for $k=10$, \Greedy1 with $t=3$ needs 242 seconds while and \Greedy2 requires only 4.1 seconds. Notice that \Greedy2 exhibits a better scalability than \Greedy1 as $k$ increases  since it requires the same time also for $k$ greater than 10.

\begin{figure}[t!]
	\begin{minipage}{0.23\textwidth}
		\centering
		\includegraphics[width=.95\linewidth]{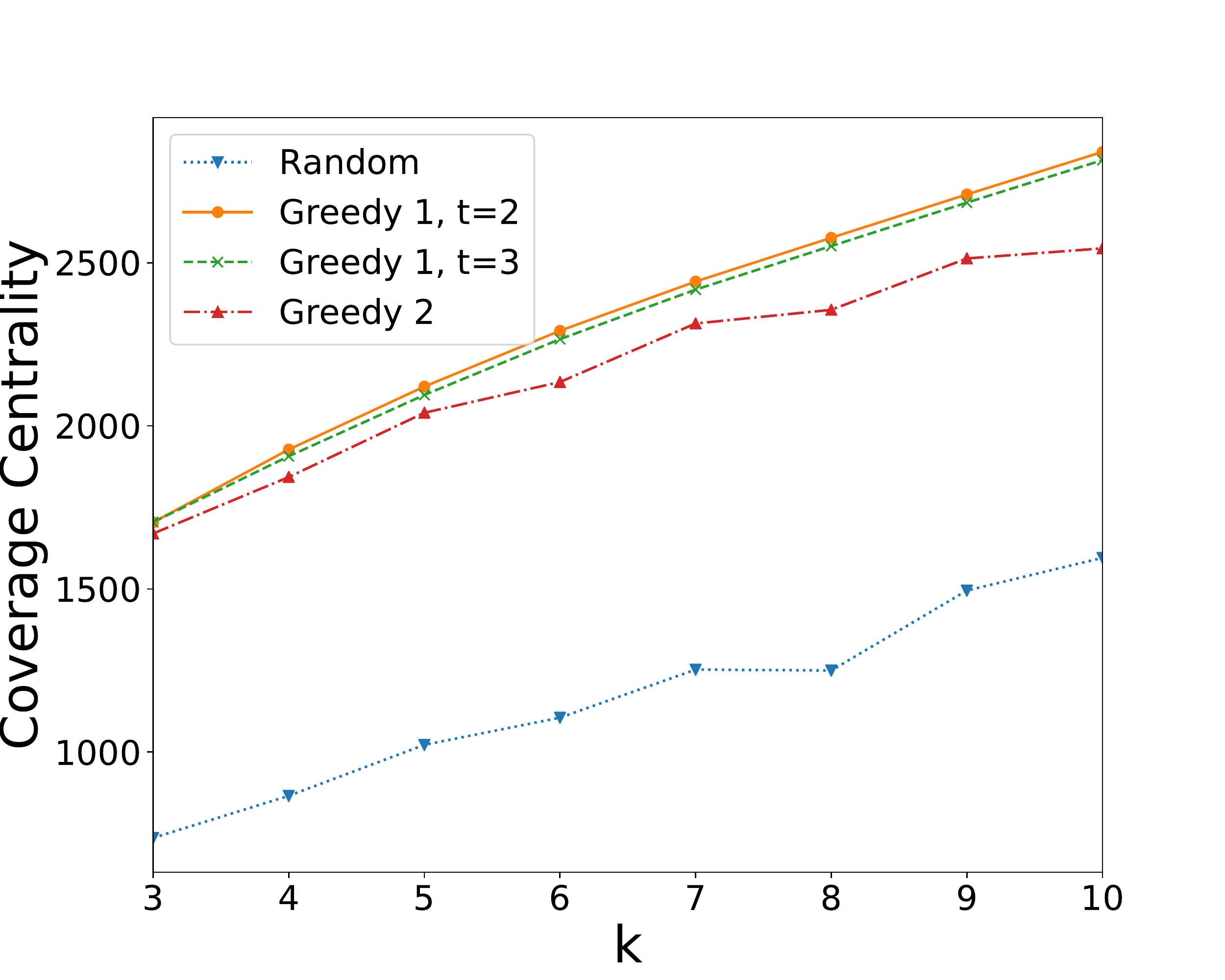}
		{\texttt{jazz} }\label{Fig:random:value:jazz}
	\end{minipage}\hfill
	\begin{minipage}{0.23\textwidth}
		\centering
		\includegraphics[width=.95\linewidth]{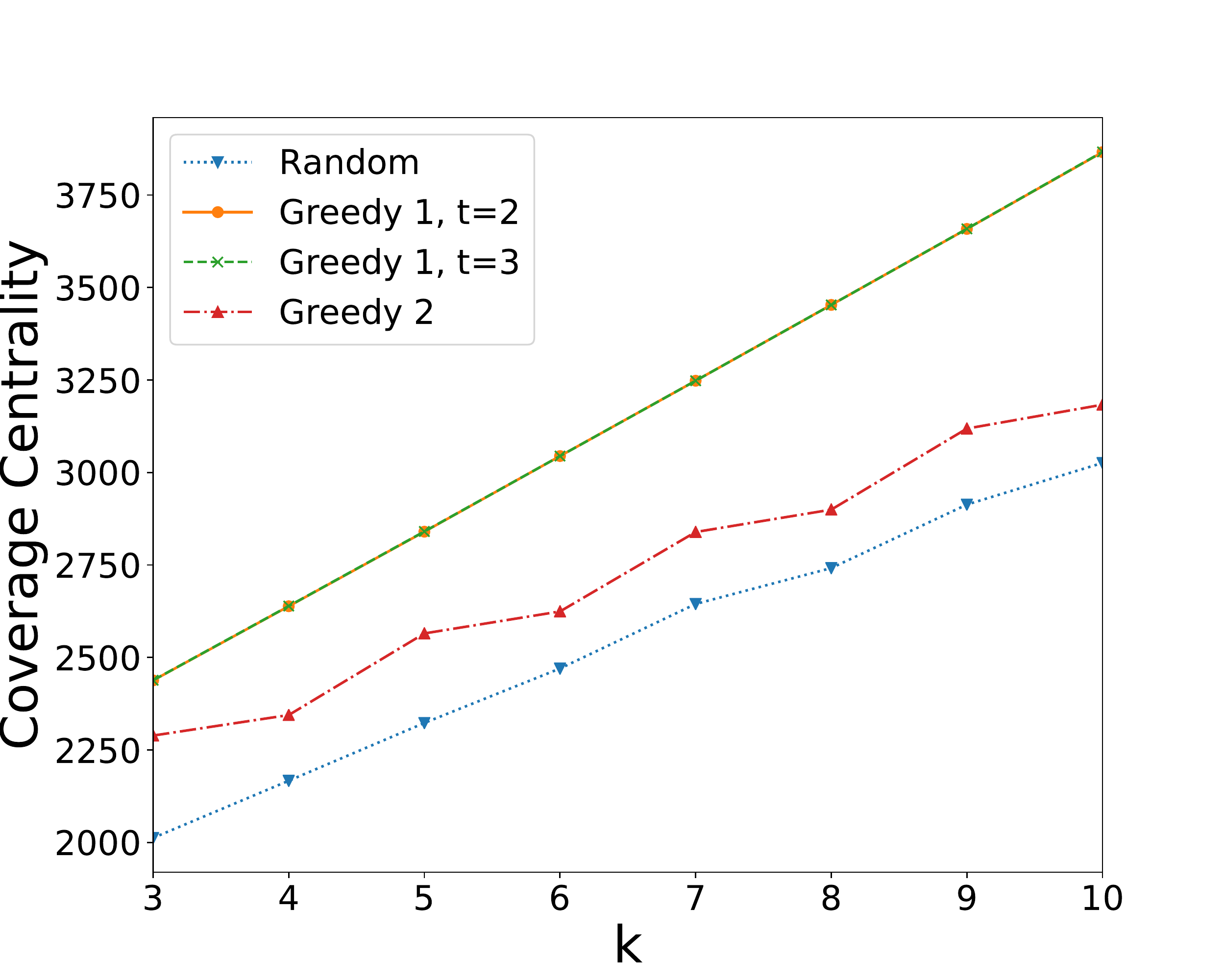}
		{\texttt{haggle }}\label{Fig:random:value:haggle}
	\end{minipage}

	\caption{Average coverage centrality of target nodes as a function of the number $k$
		of inserted edges for \Greedy1 ($t=2,3$), \Greedy 2, and \Random on \texttt{jazz} and \texttt{haggle}.}
	\label{Fig:random:value}
	\vspace{1mm}
\end{figure}

\appendix
\section{Supplementary material}
\subsection{Proof of Lemma~1}

To Prove Lemma~1, we first need to prove the following Lemma.
\setcounter{lemma}{2}
\begin{lemma}\label{lem:undirectedpairs:one}
	After each iteration $i$ of Algorithm~1, the following holds 
	\begin{equation}\label{eq:undirectedpairs:one}
	c(O_i)-c(O_{i-1}) \geq \frac{t(t-1)}{k(k-1)}( c(O^*)-c(O_{i-1})).
	\end{equation}
\end{lemma}
\begin{proof}
	For any solution $O'$ and element $x$, let us denote as $c_x(O')$ the indicator function on whether element $x$ is covered by $O'$, that is:
	\[
	c_x(O') = \left\{
	\begin{array}{ll}
	1 & \mbox{if $\exists~ i,j\in O'$ such that $x\in C(i,j)$}\\
	0 & \mbox{otherwise.}
	\end{array}
	\right.
	\]
	Therefore, the value of a solution $O'$ for \MCPP is $c(O') = \sum_{x\in X} c_x(O')$.
	
	Let $Y_t$ denote the collection of all the subsets of $O^*$ of cardinality $t$ that contains at least one element that is not a member of $O_{i-1}$: $Y_t = \{ Z \subseteq O^*: |Z| = t \wedge Z \setminus O_{i-1} \neq \emptyset \}$. We claim that the following inequality holds:
	\[
	{k-2 \choose t-2}(c(O^*)-c(O_{i-1})) \leq \sum_{Z \in Y_t} (c(O_{i-1} \cup Z) - c(O_{i-1})).
	\]
	To see that it holds, we focus on the contribution to both sides of the inequality for each element $x$ of $X$ and show the following:
	\[
	{k-2 \choose t-2}(c_x(O^*)-c_x(O_{i-1})) \leq \sum_{Z \in Y_t} (c_x(O_{i-1} \cup Z) - c_x(O_{i-1})).
	\]
	The right hand side cannot be negative so we only have to consider elements for which the left hand side is strictly positive. In other words, we consider an element $x$ in $X$ that is covered by $O^*$ but is not covered by $O_{i-1}$.
	For such an element $x$, there must be a pair $\{j,l\}$ in $Y_2$ that covers $x$, that is $x\in C(j,l)$.
	There is at least ${k-2 \choose t-2}$ members of $Y_t$ that contain $j$ and $l$ and $x$ contributes with $1$ to all the terms on the right hand side corresponding to these sets. This shows that the inequality holds.
	
	On the other hand, any subset $Z \in Y_t$ is analyzed at iteration $i$ of Algorithm~1, and a set that maximizes the difference with respect to the previous solutions is selected. This implies that 
	$c(O_{i-1} \cup Z) - c(O_{i-1}) \leq c(O_{i-1} \cup Z_i) - c(O_{i-1}) = c(O_i)-c(O_{i-1})$,
	where $Z_i$ is the set of nodes selected at iteration $i$. Since there are at most ${k \choose t}$ sets in $Y_t$, we have
	\[
	\sum_{Z \in Y_t}(c(O_{i-1} \cup Z) - c(O_{i-1})) \leq {k \choose t}( c(O_i)-c(O_{i-1})). 
	\]
	Finally, we use the fact that ${k-2 \choose t-2}/{k \choose t}=\frac{t(t-1)}{k(k-1)}$.
\end{proof}

We are now ready to prove Lemma~1.
\setcounter{lemma}{0}

\begin{lemma}\label{lem:undirectedpairs:two}
	After each iteration $i$ of Algorithm~1, the following holds 
	\begin{equation}\label{eq:undirectedpairs:two}
	c(O_i) \geq \left(1-\left(1-\frac{t(t-1)}{k(k-1)}\right)^i \right) c(O^*).
	\end{equation}
\end{lemma}
\begin{proof}
	By Lemma~\ref{lem:undirectedpairs:one}, the statement holds after the first iteration as it implies that $c(O_{1}) \geq \frac{t(t-1)}{k(k-1)}c(O^*).$ By induction, let us assume that the statement holds at iteration $i-1$, for $i\geq 2$. The following holds:
	\begin{align*}
	c(O_i) & =    c(O_{i-1}) + ( c(O_i) - c(O_{i-1}) )\\
	&\geq  c(O_{i-1}) + \frac{t(t-1)}{k(k-1)}( c(O^*)-c(O_{i-1}))\\
	& =    \left(1-\frac{t(t-1)}{k(k-1)}\right)c(O_{i-1}) + \frac{t(t-1)}{k(k-1)}c(O^*) \\
	&\geq  \left(1-\frac{t(t-1)}{k(k-1)}\right)\left(1-\left(1-\frac{t(t-1)}{k(k-1)}\right)^{i-1} \right) c(O^*)\\ &\hspace{4.5cm}+ \frac{t(t-1)}{k(k-1)}c(O^*) \\
	& = \left(1-\left(1-\frac{t(t-1)}{k(k-1)}\right)^i \right) c(O^*),
	\end{align*}
	where the first inequality is due to Lemma~\ref{lem:undirectedpairs:one} and the second one is due to the inductive hypothesis.
\end{proof}

\subsection{Proof of Lemma~2}

\setcounter{lemma}{1}
\begin{lemma}\label{lem:FKP}
	If $H$ is the set of objects selected at line~2 of Algorithm~2 and $O^*$ is an optimal solution for \MCPP, then
	\[
	d(H) \geq \frac{1}{2}\left(1-\frac{1}{e}\right)c(O^*).
	\]
\end{lemma}
\begin{proof}
	Let $H^*$ be a set of $\left\lceil\frac{k}{2}\right\rceil$ objects with maximum degree and $O_N^*$ a set of $k$ objects with maximum degree. By the approximation bound of the greedy algorithm at line~2 of Algorithm~2, it follows that
	\[
	d(H) \geq \left(1-\frac{1}{e}\right)d(H^*).
	\]
	Since $O_N^*$ and $O^*$ have the same cardinality, we have
	\[
	d(O_N^*) \geq d(O^*) \geq c(O^*).
	\]
	Let $O_N^2$ be a subset of $O_N^*$ with cardinality $\left\lceil\frac{k}{2}\right\rceil$ that maximizes $d(O_N^2)$.
	Since $O_N^2$ and $H^*$ have the same cardinality, then $d(H^*)\geq d(O_N^2)$.
	By submodularity and monotonicity of the degree function and by definition of $O_N^2$, we have
	\[
	d(O_N^*) \leq d(O_N^2) + d(O_N^*\setminus O_N^2) \leq 2 d(O_N^2),
	\]
	and then
	\[
	d(O_N^2) \geq \frac{1}{2} d(O_N^*).
	\]
	Summarizing
	\begin{align*}
	d(H)& \geq \left(1-\frac{1}{e}\right)d(H^*)  \geq \left(1-\frac{1}{e}\right)d(O_N^2)\\
	&\geq \frac{1}{2} \left(1-\frac{1}{e}\right)d(O_N^*) \geq  \frac{1}{2} \left(1-\frac{1}{e}\right) c(O^*).\qedhere
	\end{align*}
\end{proof}

\subsection{Integer program for \MCPP}
\begin{equation*}
\begin{array}{ll@{}l}
\text{Maximize}  & \displaystyle\sum_{\substack{u,w\in O\\x \in C(u,w)}} y_{xuw} &\\
\text{subject to}& \displaystyle\sum_{u,w: x \in C(u,w)}y_{xuw}\leq 1, & \mbox{ for } x \in X\\
& \displaystyle 2y_{xuw}\leq z_u + z_w,& \mbox{ for } u,w\in O, x \in C(u,w)\label{lp:mcpp}\\
& \displaystyle\sum_{u\in O} z_u \leq k,\\
& \displaystyle z_u\in\{0,1\} ,&\mbox{ for }u\in O,\\
& \displaystyle y_{xuw}\in\{0,1\} ,&\mbox{ for }u,w\in O, x\in C(u,w).
\end{array}
\end{equation*}
The binary decision variables $z_u$ and $y_{xuw}$ specify a solution $S$ of the \MCPP\ problem as follows: For any $u\in O$,
$$z_u=\begin{cases}
1&\mbox{if } u\in S,\\
0 &\mbox{otherwise,}
\end{cases}$$
and, for $u,w\in O, x \in C(u,w)$, 
$$y_{xuw} =  \begin{cases}
1&\mbox{if } \{u,w\} \mbox{ covers } x,\\
0 &\mbox{otherwise.}
\end{cases}  $$
The first constraint of the integer program ensures that every element can be covered by at most one pair $\{u,w\}$ of objects and, hence, that a covered element is counted only once in the objective function, while the second constraint ensures that an element only can be covered by certain pairs of objects.

We solve the above integer program using the Gurobi~\footnote{\url{http://www.gurobi.com/}} solver.

\newpage
\bibliography{references}
\bibliographystyle{aaai}

\end{document}